\documentclass[a4paper]{article}

\usepackage{amsthm}

\usepackage{lmodern}

\usepackage[T1]{fontenc}   
\usepackage[utf8]{inputenc}

\usepackage{algorithm}

\usepackage{algorithmic}
\usepackage{graphicx}

\newtheorem{theorem}{Theorem}
\newtheorem{lemma}{Lemma}
\newtheorem{corollary}{Corollary}
\newtheorem{definition}{Definition}
\newtheorem{fact}{Fact}

\begin{document}

\title{A QPTAS for the Base of the Number 
of Triangulations of a Planar Point Set}

\author{
Marek Karpinski
\thanks{Research partially supported by DFG grants and the Hausdorff Center grant.}\\
Department of Computer Science, University of Bonn\\
\texttt{marek@cs.uni-bonn.de}
\and
Andrzej Lingas
\thanks{Research supported in part by VR grant 621-2011-6179.}\\
Department of Computer Science, Lund University\\
\texttt{andrzej.lingas@cs.lth.se}
\and
Dzmitry Sledneu\\
Centre for Mathematical Sciences, Lund University\\
\texttt{dzmitry@maths.lth.se}
}

\maketitle

\begin{abstract}
The number of triangulations of a planar $n$ point
set is known to be $c^n$, where the base $c$
lies between $2.43$ and $30$.
The fastest known algorithm for counting triangulations
of a planar $n$ point set runs in $O^*(2^n)$ time. 
The fastest known arbitrarily close approximation
algorithm for the base of the
number of triangulations of a planar $n$ point set
runs in time subexponential in $n$.
We present the first quasi-polynomial
approximation scheme for the base of the number
of triangulations of a planar point set.
\end{abstract}

\section{Introduction}

A \emph{triangulation}  $T$ of a set $S$ of $n$ points
in the Euclidean plane is a maximal set of
properly non-intersecting straight-line segments
with both endpoints in $S$. These straight-line segments
are called \emph{edges} of $T$.
Let $F(S)$ stand for the set of all triangulations
of $S$. 

The problem of computing the number of triangulations
of $S$, i.e., $|F(S)|$, is easy when $S$ is convex.
Simply, by a straightforward recurrence, 
$|F(S)|=C_{n-2}$, where $C_k$ is the $k$-th Catalan
number, in this special case. However, in the general
case, the problem of computing the number
of triangulations of $S$ is neither known to
be $\# P$-hard nor known to admit a polynomial-time counting algorithm.

It is known that $|F(S)|$ lies between $\Omega (2.43^n)$ \cite{SSW11}
and $O(30^n)$ \cite{SS11}. Since the so called flip graph
whose nodes are triangulations of
$S$ is connected~\cite{S78}, all triangulations of
$S$ can be listed in exponential time by a standard
traversal of this graph. 
Only recently, Alvarez and Seidel have presented an elegant
algorithm for the number of triangulations of
$S$ running in $O^*(2^n)$ time~\cite{AS13}
which is substantially below the aforementioned
lower bound on $|F(S)|$ (the $O^*$ notation suppresses polynomial
in $n$ factors).

Also recently, Alvarez, Bringmann, Ray, and Seidel~\cite{ABRS13}
have presented
an approximation algorithm for the number
of triangulations of $S$ based on a recursive
application of the planar simple cycle separator~\cite{M86}.
Their algorithm runs in subexponential $2^{O(\sqrt n\log n)}$ time
and over-counts the number of
triangulations by at most a subexponential 
$2^{O(n^{\frac 34}\sqrt{\log n})}$ factor.
It also yields a subexponential-time 
approximation scheme for the base
of the number of triangulations
of $S$, i.e., for $|F(S)|^{\frac 1n}$.
The authors of~\cite{ABRS13} observe also that
just the inequalities $\Omega (2.43^n)\le |F(S)|\le O(30^n)$   
yield the large exponential approximation 
factor $O(\sqrt {30/2.43}^n)$ 
for $|F(S)|$ trivially computable
in polynomial time.

\subsection{Our contribution}

We take a similar approximation approach to the
problem of counting triangulations of $S$
as Alvarez, Bringmann, Ray, and Seidel in~\cite{ABRS13}. However, importantly,
instead of using recursively the planar simple cycle separator~\cite{M86}, 
we apply recursively the so called
balanced $\alpha$-cheap $l$-cuts 
of maximum independent sets of triangles within
a dynamic programming framework developed
by Adamaszek and Wiese in~\cite{AW14,AW13}. 
By using the aforementioned techniques,
the authors of~\cite{AW14} designed 
the first quasi-polynomial time approximation
scheme (QPTAS) for the maximum weight independent set of
polygons belonging to the input set
of polygons with poly-logarithmically
many edges.

Observe that a triangulation of $S$ can
be viewed as a maximum independent set
of triangles drawn from the set of
all triangles with vertices in $S$ 
that are free from other points in $S$
(triangles, or in general polygons, are 
identified with their open interiors).
This simple observation enables us
to use the aforementioned
balanced $\alpha$-cheap $l$-cuts 
recursively in order to bound
an approximation factor of our approximation
algorithm. 
The parameter $\alpha$
specifies the maximum fraction of an
independent set of triangles that can be
destroyed by the $l$-cut, which is a polygon 
with at most $l$ vertices in a specially constructed
set of points of polynomial size.
%The other parameter $l$ specifies
%the length of the cut.

Similarly as the approximation algorithm
from~\cite{ABRS13}, our algorithm may over-count
the true number of triangulations
because the same triangulation can
be be partitioned recursively in many different
ways. In contrast with the approximation algorithm in~\cite{ABRS13},
our  algorithm may also under-count the
number of triangulations of $S$, since our
partitions generally destroy a fraction
of triangles in a complete triangulation of $S$.

Our approximation
algorithm for 
the number of triangulations of a set $S$
of $n$ points with integer coordinates in the plane
runs 
in $n^{(\log (n) /\epsilon)^{O(1)}}$ time.
For $\epsilon > 0$, it returns
a number at most $2^{\epsilon n}$ times
smaller and at most $2^{\epsilon n}$ times
larger than the number of triangulations
of $S$. Note that even for 
$\epsilon = (\log n )^{-O(1)}$,
the running time is still quasi-polynomial.

As a corollary, we obtain
 a quasi-polynomial approximation
scheme for the base of the number of triangulations
of $S$, i.e., for $|F(S)|^{\frac 1n}$.
This implies that the problem of
approximating $|F(S)|^{\frac 1n}$ cannot be
APX-hard (under standard complexity theoretical
assumptions).

\subsection{Organization of the paper}
In Preliminaries, we introduce 
basic concepts of the dynamic programming
framework from~\cite{AW13}. Section 3 presents
our approximation counting algorithm for the number
of triangulations of $S$. In Section 4,
the time complexity of the algorithm is
examined while in Sections 5, upper bounds
on the under-counting and the over-counting
of the algorithm are derived, respectively.
In Section 6, our main results are formulated.
We conclude with a short discussion on
possible extensions of our results in Section 7.

\section{Preliminaries}

The Maximum Weight Independent Set of Polygons Problem
(MWISP) is defined as follows~\cite{AW14}.
We are given a set ${Q}$  of $n$ polygons in the Euclidean plane.
Each polygon has at most $k$ vertices, each of the vertices
has integer coordinates. Next, each polygon $P$ in $Q$
is considered
as an open set, i.e., it is identified with the set
of points forming its interior. Also, each polygon 
$P\in Q$ has weight $w(P)>0$ associated with it. The task
is to find a maximum weight independent set of polygons
in $Q$, i.e., a maximum weight set $Q'\subseteq Q$ such that for all
pairs $P_i,\ P_j$ of polygons in $Q'$, if $P_i \neq P_j$
then it holds $P_i \cap P_j = \emptyset $.

The \emph{bounding box} of $Q$ is the smallest
axis aligned rectangle containing all polygons
in $Q$.

Note that in particular if 
${Q}$ consists of all triangles with vertices
in a finite planar point set $S$ such that
no other point in $S$ lies inside them or
on their perimeter, each having weight $1$,
then the set of all maximum independent
sets of polygons in ${Q}$ is just the set
of all triangulations of $S$. We shall denote
the latter set by $F(S)$.

Adamaszek and Wiese have shown that if $k=poly(\log n)$
then MWISP admits a QPTAS~\cite{AW14}.
\begin{fact}[\cite{AW14}] Let $k$ be a positive integer.
There exists a $(1+\epsilon)$-approximation
algorithm with a running time of $(nk)^{({\frac k{\epsilon}}\log
    n)^{O(1)}}$
for the Maximum Weight Independent Set of Polygons Problem provided
that each polygon has at most $k$ vertices.
\end{fact}
Recently, Har-Peled generalized Fact 1 to include arbitrary
polygons~\cite{H14}.

We need the following tool from  \cite{AW14}.

\begin{definition}
Let $l\in N$ and $\alpha \in R$ where $0<\alpha<1$.
Let $T$ be a set of pairwise non-touching triangles.
A polygon $\Gamma$  is a balanced $\alpha$-cheap
$l$-cut of $T$ if
\begin{itemize}
\item $\Gamma$ has at most $l$ edges,
\item the total weight of all triangles
in $T$ that intersect $\Gamma$
does not exceed an $\alpha$ fraction
of the total weight of triangles in $T$,
\item 
the total weight of the triangles in $T$
contained in $\Gamma$ does not
exceed two thirds of the total weight
of triangles in $T$,
\item 
the total weight of the triangles in $T$
outside $\Gamma$ does not
exceed  two thirds of the total weight
of triangles in $T$.
\end{itemize}
\end{definition}

For a set of triangles $T$ in the plane,
the set of \emph{DP-points} consists
of \emph{basic DP-points} and \emph{additional DP-points}.
The set of basic DP-points contains
the four vertices of the bounding box
of $T$ and each intersection of
a vertical line passing through a corner
of a triangle in $T$ with any edge
of a triangle in $T$ or a horizontal
edge of the bounding box.
The set of additional DP-points 
consists of all intersections of
pairs of straight-line segments whose
endpoints are basic DP-points.
The authors of~\cite{AW14} observe that
the total number of DP-points is $O(n^4)$.

\begin{fact}[\cite{AW14}] Let $\delta >0$ 
and let $T$ be a set of pairwise non-touching triangles
in the plane such that the weight of no triangle in $T$
exceeds one third of the weight of $T$. Then there exists
a balanced $O(\delta)$-cheap $(\frac 1{\delta} )^{O(1)}$-cut 
with vertices at basic DP-points.
\end{fact}

By a maximal triangulation of input points within a DP-cell,
we shall mean a partial
triangulations of these points
that cannot be extended by any edge within the cell.
Next, by a maximal fragment of a triangulation $R$ of input points, within a
DP cell, we shall mean the partial triangulation consisting of all
edges of $R$ that are contained  in the cell.

\section{Dynamic programming}

The QPTAS of Adamaszek and Wiese
for maximum weight independent set of polygons~\cite{AW14} is
based on dynamic programming. For each polygon
with at most $k$ vertices at the 
DP points induced by the input polygons, termed a DP cell, an approximate
maximum weight independent subset of the input polygons contained
in the DP cell is computed. The computation is done
by considering all possible partitions of the DP cell into
at most $k$ smaller DP cells. For each such partition,
the union of the approximate solutions for the component
DP cells is computed. Then, a maximum weight union
is picked as the approximate solution for the DP cell.

Our  dynamic programming algorithm, termed Algorithm 1
and depicted in Fig. 1.,
is in part similar to that of Adamaszek and Wiese~\cite{AW14}.
The set of input polygons consists
of all triangles with three vertices
in the input planar point set $S$
that do not contain any other
point in $S$. 
For each DP cell, an approximate
number of triangulations
within the cell is computed
instead of an approximate
maximum number of non-touching triangles
within the cell.
Further modification
of the dynamic programming
of Adamaszek and Wiese are as follows.
\begin{enumerate}
\item
Solely those partitions of a DP cell
into at most $k$ component DP cells
are considered where no component cell contains
more than two thirds of the
input points 
(i.e., the vertices
of the input triangles) 
in the partitioned cell.
(Alternatively, one could
generalize the concept of a DP cell
to a set of polygons with holes
and consider only partitions into two
DP cells obeying this restriction.)
\item
While a partition of a DP cell
into at most $k$ cells is 
processed, instead of the
union of the solutions
to the subproblems for
these cells, the product
of the numerical solutions
for the component DP cells
is computed. 
\item
Instead of taking the maximum
of the solutions induced
by the partition of a DP cell into at
most $k$ DP-cells, 
the sum of the numerical solutions 
induced by these partitions is computed.
\item
When the number of points
contained in a DP cell 
does not exceed the threshold number
$\Delta$ then 
the exact number of maximal triangulations
within the cell is computed.
\end{enumerate}

\begin{figure}
\begin{algorithmic}[1]
\REQUIRE A set $S$ of $n$ points with integer
coordinates in the Euclidean plane
and natural number parameters $k$ and $\Delta$. 
\ENSURE  An approximate number of triangulations of $S$. 

\STATE $T\leftarrow$
the set of all triangles
with vertices in $S$
that do not contain 
any other point in $S$;

\STATE $P\leftarrow $ a list 
of polygons (possibly with holes)
with at most $k$ vertices in total
at DP points induced by $T$,
topologically sorted 
with respect to geometric containment;
\FOR {each polygon $Q\in P$ 
containing at most $\Delta$ points in
$S$}

\STATE $tr(Q)\leftarrow$ exact number
of maximal triangulations
of points in $S\cap Q$ 
within $Q$;
\ENDFOR
\FOR {each polygon set $Q\in P$ 
containing more than $\Delta$ points in
$S$} 
\STATE $tr(Q)\leftarrow 0$;
\FOR {each partition of $Q$
into polygons $Q_1,\dots,Q_l \in P$,
where $l\le k$, no $Q_j$ contains
more than two thirds of points in $S\cap Q$, and 
$tr(Q_1)$ through $tr(Q_l)$ are defined}
%into two polygon sets $Q_1,\ Q_2 \in Q$
%by a balanced $\frac 23$-cheap cut
%with DP-vertices 
%(with respect to the triangles in $T$
%contained in $P$), where $tr(Q_1)$ and $tr(Q_2)$
%are defined} 
\STATE $tr(Q)\leftarrow tr(Q)+\prod_{j=1}^ltr(Q_j)$;
%\STATE $tr(Q)\leftarrow tr(Q)+tr(Q_1)tr(Q_2)$;
\ENDFOR
\ENDFOR
\STATE Output $tr(B)$, where $B$ is
the bounding box of $T$.
\end{algorithmic}
\caption{Algorithm 1 for approximately counting
triangulations of a finite planar point set.}
\label{fig: algo1}
\end{figure}

Algorithm 1 also in part resembles the approximation counting algorithm
for the number of triangulations of a planar point
set due to Alvarez, Bringmann, Ray, and Seidel~\cite{ABRS13}.
The main difference is in the used implicit recursive partition
tool. Algorithm 1 uses balanced $\alpha$-cheap $l$-cuts
within the dynamic programming framework from~\cite{AW13}
instead of the simple cycle planar
separator theorem~\cite{ABRS13,M86}.
Thus, Algorithm 1 recursively partitions 
a DP cell defining a subproblem into
at most $k$ smaller DP-cells while the
algorithm in~\cite{ABRS13} recursively
splits a subproblem by a simple cycle
that yields a balanced partition. 
The
new partition tool gives a better running time since the number of
possible partitions is much smaller so the dynamic
programming/recursion has lower complexity and the threshold for the
base case can be much lower.  Since the algorithm in \cite{ABRS13} in
particular lists all simple cycles on $O(\sqrt n)$ vertices, it runs
in at least $2^{O(\sqrt n \log n)}$ time independently of the
precision of the approximation.

\section{Time complexity}

The cardinality of $T$ does not exceed $n^3$.
Then, by the analogy with the dynamic
programming algorithm of Adamaszek and Wiese for
nearly maximum independent set
of triangles, the number
of DP cells is $(3n^3)^{O(k)}=n^{O(k)}$
(see Proposition 2.1 in~\cite{AW14}).
Consequently, the number of possible
partitions of a DP cell into at most
$k$ DP cells is $O({n^{O(k)} \choose k})$
which is  $n^{O(k^2)}$.

It follows that if we
neglect the cost of computing
the exact number of triangulations
contained within a DP cell including
at most $\Delta$ input points, then
Algorithm 1
runs in $n^{O(k^2)}$ time.
% as that of AW.

We can compute the exact
number of maximal triangulations contained
within a DP cell with at most
$\Delta$ input points   by
listing all complete triangulations
of these points and counting
only their maximal fragments
within the cell that are
maximal triangulations within the cell.
To list 
the complete triangulations
of the input points within a DP cell
we can apply any traversal
algorithm, like BFS or DFS, to
the so called {\em flip} graph
 of the triangulations which
is known to be connected \cite{S78}.
Hence, the listing takes
time proportional
to the number of
the triangulations, which
is at most $30^{\Delta}$ by \cite{SS11}.
It follows that
extracting the maximal fragments
and counting those
being maximal triangulations can be done
also in $2^{O(\Delta)}n^{O(1)}$ time.

We conclude with the following lemma.

\begin{lemma}\label{lem: time}
Algorithm 1 runs
in  $n^{O(k^2)}2^{O(\Delta)}$ time. 
\end{lemma}

\section{Approximation factor}

\subsection{Under-counting}

The potential under-counting
stems from the fact that when
a DP cell is partitioned
into at most $k$ smaller DP cells
%a balanced cut is applied
%to a triangulations within
%a DP cell 
then the possible
combinations of triangulation
edges crossing the partitioning edges
are not counted. Furthermore,
in the leaf DP cells, i.e.,
those including at most $\Delta $
points from $S$, we count only
maximal triangulations while
the restriction of a triangulation
of $S$ to a DP cell does not have
to be maximal.
See Fig.~\ref{fig. 2}.

\begin{figure}
\centering
\includegraphics[scale=0.4]{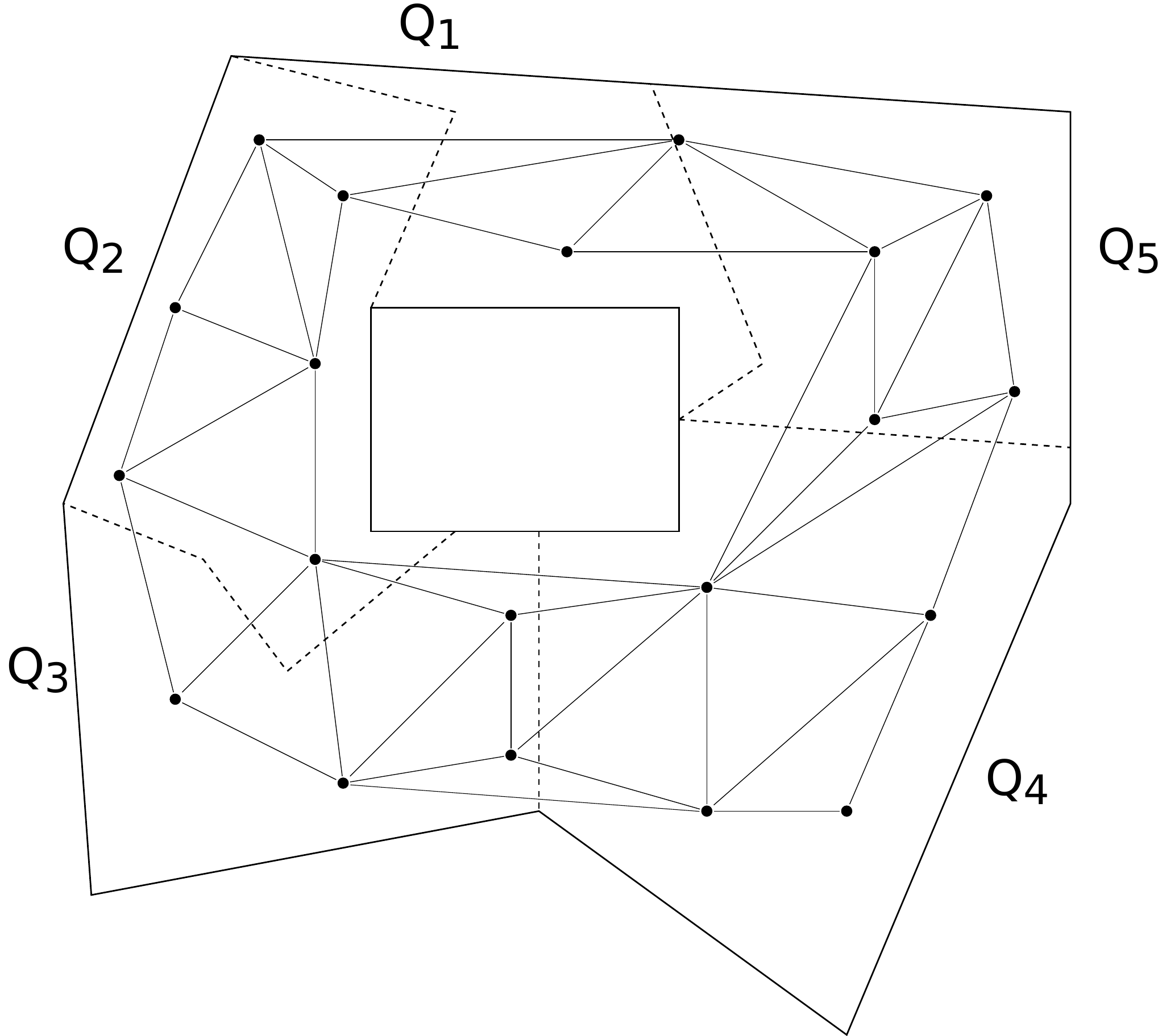}
\caption{An example of a maximal triangulation
within a DP cell and a partition of the DP cell into
smaller DP cells $Q_1$,\dots,$Q_5$ crossing some
triangles in the triangulation.}

\label{fig. 2}
\end{figure}

Intuitively, the general idea of the proof
of our upper bound on under-counting is
as follows. For each triangulation $W\in F(S)$,
there is triangulation counted by Algorithm 1
that can be obtained by removal $O(\epsilon n)$
edges and augmenting with $O(\epsilon n)$
other edges. Such a triangulation
is a union of maximal triangulations
contained in leaf DP cells.

\begin{lemma}\label{lem: under}
Let $S$ be a set of $n$ points in the plane
and let $\epsilon >0$.
For each $W\in F(S)$, there is a partial
triangulation $W^*\subseteq W$ of $S$
containing at least a $1-O(\epsilon)$ 
fraction of the triangular faces
of $W$ and a partial triangulation
$M(W^*)$ of $S$ which is an extension
of $W^*$ by $O(\epsilon n)$ edges
such that
the estimation returned by
Algorithm 1 with $k$ set to
$\log^{O(1)} (n)/\epsilon^{O(1)}$ is not less
than  $|\bigcup_{W\in F(S)} \{ M(W^*)\}|$.
\end{lemma}

\begin{proof}
Let $W\in F(S)$. By adapting the idea
of the proof of the approximation ratio
of the QPTAS in~\cite{AW14}, consider the following
tree $U$ of DP cells obtained by recursive
applications of balanced $\alpha$-cheap $l$-cuts.

At the root of $U$, there is the bounding box.
By Fact 2, there is a balanced $\alpha$-cheap $l$-cut,
where $l=\alpha^{-O(1)}$,
that splits
the box into at most $k$ children DP cells such that only
$\alpha $ fraction of the triangular faces
of $W$ is crossed by the cut. The construction
of $U$ proceeds recursively in children DP cells
and stops in DP cells that contain at most
$\Delta$ points.

Note that the height of $U$
is not greater than $\log_{3/2}n $.

For a node $u$ of $U$, let $W_u$ 
be the partial triangulation
of the points in the DP cell $Q_u$ associated
with $u$ that is a restriction of $W$
%(i.e., its sets of vertices and edges are subsets of those of $W$) 
to (the vertices and edges of $W$ contained in) $Q_u$.
Next, let $W_u^*$ denote
the restriction of $W_u$ to
the union of $W_t$ over the 
the leaves $t$ of the subtree of $U$ rooted
at $u$.

By induction on the height $h(u)$
of $u$ in $U$, we obtain that
the partial triangulation $W^*_u\subseteq W_u$
contains $(1-\alpha )^{h(u)}$ fraction
of triangular faces of $W_u$.
Set $\alpha $ to $\frac {O(\epsilon ) }{\log (n/\epsilon)}$.
It follows in particular that
for the root $r$ of $U$,
$W_r^*\subseteq W$ contains
at least an $(1-\alpha )^{\log_{3/2}n/{\epsilon}}\ge 
1-O(\epsilon)$ fraction of triangular faces
in $W$. Set $W^*$ to $W_r^*$.
%Again by the induction on the height of a node
%$u$ of $U$, we easily obtain that 
%$W_u^*$ is a union of $W^*_t$ over
%the leaves of the subtree of $U$ rooted
%at $u$. 

For a leaf $t$ of $U$, let $M(W_t)$
be an extension of $W_t$ to a maximal
triangulation within the leaf cell $D_t$.
For a node $u$ of $U$, let $M(W^*_u)$
be a partial triangulation within $Q_u$
that is the union of $M(W_t)$
over the leaves $t$ of the subtree of $U$ rooted
at $u$. We have also $M(W^*)=M(W^*_r)$
by $W^*=W_r^*$. 

By the definition of $M(W^*)$,
$M(W^*)$ is an extension of $W^*$.
By the definition of $W^*$,
any edge of $M(W^*)$ that
is not an edge of $W^*$
has both endpoints at vertices
of triangles in $W$ that are
missing in $W^*$. It follows
that the number of edges
in $M(W^*)\setminus W^*$
is at most
$3\times O(\epsilon n)=O(\epsilon n)$.

We shall show by induction on $h(u)$
that Algorithm 1 while computing
an estimation for $Q_u$ 
%approximately 
at least counts 
the number of $M(W_u^*)$.

If $h(u)=0$, i.e., $u$ is a leaf
in $U$ then $W_u^*=W_u$ and consequently 
in particular $M(W_u^*)=M(W_u)$
is counted by  Algorithm 1.

Suppose in turn that $u$ is an internal
node in $U$ with $l$ children $u_1,\dots,u_l$.
When the estimation 
for $Q_u$ is computed by Algorithm 1, the sum
of products of estimations yielded by
different partitions of $Q_u$ into at most
$k$ DP cells is computed. In particular,
the partition into $Q_{u_1},\dots,Q_{u_l}$ is
considered. By the induction
hypothesis, the estimation for
$Q_{u_j}$ includes $M(W^*_{u_j})$ 
for $j=1,\dots,l$. Hence,
the product of these estimations
counts also $M(W_u^*)=\bigcup^l_{j=1}M(W^*_{u_j})$. 

By $M(W^*)=M(W_r^*)$, to obtain the lemma 
it remains to show that the bound
$\log^{O(1)} (n/\epsilon)/\epsilon^{O(1)}$
on $k$ is sufficiently large. 
Following the proof of Lemma 2.1 in~\cite{AW14},
observe that each DP cell $Q_u$ at each level
of $U$ is an intersection of at most
$O(\log (n/\epsilon ))$ polygons, each with
at most $l$ edges and vertices at basic DP points.
Hence, by $\alpha = \frac {O(\epsilon )}{\log (n/\epsilon)}$
and $l=\alpha^{-O(1)}$,
the resulting polygons have at
most $O(l^2\log^2 (n/\epsilon))$
$=\log^{O(1)} (n/\epsilon)/\epsilon^{O(1)}$
edges and vertices at basic and additional
DP points.
\end{proof}

\begin{theorem} \label{theo: under}
The under-counting factor of
%the algorithm for approximate
%triangulation counting 
Algorithm 1 with
$k$ set to\\
$\log^{O(1)} (n/\epsilon)/\epsilon^{O(1)}$ is
at most $2^{O(\epsilon n\log n)}$.
\end{theorem}
\begin{proof}
Consider any triangulation $W\in F(S)$.
By Lemma~\ref{lem: under}, the partial triangulation $W^*\subseteq W$
contains at least an $1-O(\epsilon)$ fraction
of the triangular faces of $W$. Hence, the edges
completing $W^*$ to any triangulation in $F(S)$
can be incident only to the vertices of the
triangular faces in $W$ that do not occur
in $W^*$. The number of the latter is at most
$3\times O(\epsilon n)=O(\epsilon n)$. It follows that
the number of ways of completing $W^*$ to a
full triangulation in $F(S)$ is not greater
than the number of full triangulations
on $O(\epsilon n)$ vertices which is
not greater than $30^{O(\epsilon n)}=2^{O(\epsilon n)}$.

By Lemma~\ref{lem: under}, 
the estimation returned by
Algorithm 1 with $k$ set to\\
$\log^{O(1)} (n/\epsilon)/\epsilon^{O(1)}$ is not less
than  $|\bigcup_{W\in F(S)} \{ M(W^*)\}|$.

Now it remains to show that the maximum number of
partial triangulations $(W')^*$, $W'\in F(S)$, for which
$M((W')^*)=M(W^*)$
is at most $2^{O(\epsilon n\log n)}$. To see this
observe that the edges extending $(W')^*$
to $M((W')^*)$ are incident to at most $O(\epsilon n)$
vertices of the $O(\epsilon n)$ triangular faces of $W'$
that are missing in $(W')^*$. Consequently,
the maximum number of such partial triangulations $(W')^*$
is upper bounded by the number of subsets
of at most $O(\epsilon n)$ edges of $M(W^*)$
(whose removal may form a partial triangulation
$(W')^*$ satisfying $M((W')^*)=M(W^*)$ ).
The latter number is $2^{O(\epsilon n \log n)}$.
%Now the theorem follows by Lemma~\ref{lem: under}.

We conclude that for $W\in F(S)$, 
the number of other triangulations
$W'\in F(S)$ for which $M((W')^*)=M(W^*)$
is at most $2^{O(\epsilon n)}2^{O(\epsilon n\log n)}=
2^{O(\epsilon n\log n)}$. Now, the theorem follows
from Lemma~\ref{lem: under}.
\end{proof}

\subsection{Over-counting}

The reason for over-counting in the estimation
returned by our algorithm is as follows.
The same
triangulation within a DP cell
may be cut in the number of
ways proportional to the number 
of considered partitions of the DP cell into at most $k$
smaller DP-cells. This reason is similar to
that for over-counting of the approximation
triangulation counting algorithm of  Alvarez, Bringmann, Ray, 
and Seidel~\cite{ABRS13}
based on the planar simple cycle separator theorem.
Therefore, 
our initial recurrences and calculations 
are similar to those derived
in the analysis of the over-counting from~\cite{ABRS13}.

\begin{lemma}\label{lem: over1}
Let $Q$ be an arbitrary DP cell 
processed by Algorithm 1 which
contains more than $\Delta$ input points.
Recall the calculation of the estimation for
$Q$ by summing the products of estimations
for smaller DP cells $Q_1,\dots,Q_l$ over $n^{O(k^2)}$ partitions
of $Q$ into $Q_1,\dots,Q_l$, $l\le k$.
Substitute the true value of
the number of maximal triangulations within 
each such smaller cell $Q_i$ for the estimated one
in the calculation.
Let $r$ be the resulting value. The
number of maximal triangulations within
$Q$ is at least $r/n^{O(k^2)}$.
\end{lemma}
\begin{proof}
Note that $r$ is the sum of the number
of different combinations  
of maximal triangulations within  
smaller DP cells $Q_1,\dots,Q_l$
over $n^{O(k^2)}$ partitions of
$Q$ into smaller cells $Q_1,\dots,Q_l$, 
$l\le k$. Importantly, each
such combination can be completed
to some maximal triangulation within $Q$
but no two different combinations coming
from the same partition $Q_1,\dots,Q_l$
can be extended to the same maximal triangulation
within $Q$.

Let $M$ be the set of maximal triangulations $W$
within $Q$ for which there is a partition
into smaller DP cells $Q_1,\dots,Q_l$, $l\le k$,
such that for $i=1,\dots,l$, $W$ constrained
to $Q_i$ is a maximal triangulation within $Q_i$.
Note that for each $W\in M$, the number of the combinations
that can be completed to $W$  cannot exceed that
of the considered partitions, i.e., $n^{O(k^2)}$,
as each of the combinations has to come
from a distinct partition $Q_1,\dots,Q_l$.

Thus, there is a binary relationship between
maximal triangulations within $Q$ that
belong to $M$ and the aforementioned
combinations. It is defined on all the
maximal triangulations in $M$ and on all the combinations,
and a maximal triangulation in $M$
is in relation with at most $n^{O(k^2)}$ combinations.
This yields the lemma.
\end{proof}

By Lemma~\ref{lem: over1}, we can express the over-counting
factor $L(Q,\Delta)$ of Algorithm 1 for 
a DP cell $Q$ 
%with $m>\Delta$ input points 
by the following recurrence:

$$L(Q,\Delta)=\sum_{(Q_1,\dots,Q_l)} \prod_{j=1}^lL(Q_j,\Delta)\le
n^{O(k^2)} \prod_{j=1}^{l^*}L(Q_j^*,\Delta)$$

where the summation is over all partitions of $Q$ into
DP cells $Q_1,\dots,Q_l$, where $l\le k$, 
and $Q^*_1,\dots,Q^*_{l^*}$ is a partition
that maximizes the term $\prod_{j=1}^lL(Q_j,\Delta)$.
When $Q$ contains at most $\Delta$ input points,
Algorithm 1 computes the exact number of maximal
triangulations of these points within $Q$. Thus, we have
$L(Q,\Delta)=1$  in this case.
%when $m\le \Delta $.

Following~\cite{ABRS13}, it will be more convenient
to transform our recurrence by taking
logarithm of both sides.
For any DP cell $P$, let
$L'(P,\Delta)=\log L(P,\Delta)$.
%and $L'(P,\Delta)=\log L(P,\Delta)$.
We obtain now:

$$L'(Q,\Delta)\le
O(k^2\log n ) +\sum_{j=1}^{l^*}L'(Q_j^*,\Delta)$$

\begin{lemma}\label{lem: over2}
Let $B$ be a bounding box for a set $S$ of
$n$ points in the plane.
% and $\Delta$
%is larger than sufficiently large constant. 
%Suppose that the function $k(\ )$ is not decreasing.
The following equality holds
$$L'(B,\Delta)=O(k^2n\log^2 n/\Delta )$$
\end{lemma}
\begin{proof}
Let $U$ be the recurrence tree and let $D$
be the set of direct ancestors of leaves
in $U$. For each node $d\in D$, the corresponding
$DP$ cell includes at least $\Delta +1$ 
points in $S$. It follows that $|D|\le n/\Delta $.
Also, any node in $D$ has depth $O(\log n)$  in $U$.
Consequently, the contribution of the subproblems
corresponding to  nodes in $D$
and their ancestors to the estimation for $L'(B,\Delta)$
can be upper bounded by $O(k^2\log n \times (n/\Delta)\log n)$.
Finally, recall that the subproblems corresponding
to leaves of $U$ do not contribute 
to the estimation.
\end{proof}

Lemma~\ref{lem: over1} 
immediately yields
the following corollary.

\begin{theorem}\label{theo: over}
Let $B$ be a bounding box for a set of
$n$ points in the plane.
Set the parameter $k$ in Algorithm 1
as in Theorem~\ref{theo: under}.
If for $\epsilon >0$
the parameter $\Delta$ in Algorithm 1
is set to $\frac c{\epsilon}k^2\log^2 n$
for sufficiently large constant $c$ then
the over-counting factor is at most
$2^{\epsilon n}$.
\end{theorem}

\section{Main result}

By combining Lemma~\ref{lem: time}
with Theorems~\ref{theo: under},~\ref{theo: over} with $\epsilon$
set to $\epsilon /\log n$, we obtain
our main result.

\begin{theorem}\label{theo: main}
There exists an approximation
algorithm for 
the number of triangulations of a set $S$
of $n$ points with integer coordinates in the plane
with a running time of at most
$n^{(\log (n) /\epsilon)^{O(1)}}$  that returns
a number at most $2^{\epsilon n}$ times
smaller and at most $2^{\epsilon n}$ times
larger than the number of triangulations
of $S$.
\end{theorem}

\begin{corollary}
There exists a $(1 +\epsilon)$-approximation
algorithm with a running time of at most
$n^{(\log (n) /\epsilon)^{O(1)}}$ for the base
of the number of triangulations of a set
of $n$ points with integer coordinates in the plane.
\end{corollary}
\begin{proof}
Let $c^n$ be the number of triangulations
of the input $n$ point set, and let $\Lambda$
be the number returned by the algorithm
from Theorem~\ref{theo: main}.
We have $\max \{ \frac {c^n} {\Lambda}, \frac {\Lambda} {c^n} \} \le 2^{\epsilon n}$
by Theorem~\ref{theo: main}.
By taking the $n$-th root on both sides, we obtain
$\max \{ \frac {c} {\Lambda^{\frac 1n}}, \frac {\Lambda^{\frac 1n}} {c} \} \le
2^{\epsilon }$.
%$2^{-\epsilon }c\le \Lambda^{\frac 1n} \le 2^{\epsilon}c$.
Now it is sufficient to observe that
$2^{\epsilon} < 1+\epsilon $ for $\epsilon < \frac 12$.
\end{proof}

\section{Extensions}

Adamaszek and Wiese presented also an extension
of their theorem on $\alpha$-cheap cut of an independent
set of triangles (Fact 2) to include independent
polygons with at most $K$ edges (Lemma 3.1~\cite{AW14}). This makes possible
to generalize our QPTAS
to include the approximation of the number
of maximum weight partitions into $K$-gons.

\section*{Acknowledgments}

We thank Victor Alvarez, Artur Czumaj, Peter Floderus, Miroslaw
Kowaluk, Christos Levcopoulos and Mia Persson for preliminary
discussions on counting the number of triangulations of a planar point
set. We are also very grateful to unknown referees for pointing out an
imprecision in the under-counting analysis in a preliminary version of
our paper as well as and many other valuable comments.

\end{document}